\documentclass[
 reprint,
 superscriptaddress,
nofootinbib,
amsmath,amssymb,
aps,
pra,
twocolumn,
longbibliography
]{revtex4-1}

\usepackage{amsthm}
\usepackage{amsmath,amsfonts,amssymb}
\usepackage{mathtools}
\usepackage{graphicx,xcolor}% Include figure files
\usepackage{dcolumn}% Align table columns on decimal point
\usepackage{bm}% bold math
\usepackage{bbold}
\usepackage{physics} % enables ket and bra notation
\usepackage{natbib}
\usepackage{hyperref}% add hypertext capabilities
\usepackage{url}

\usepackage[normalem]{ulem}

\providecommand{\ignore}[1]{}

\newif\ifcmnt
\cmnttrue
\ifdefined\cmntsoff\cmntfalse\fi

\ifcmnt
    \providecommand{\aucmnt}[1]{#1}

\else
    \providecommand{\aucmnt}[1]{}

\fi

\newtheorem{proposition}{Proposition}
\newtheorem{theorem}{Theorem}

\newcommand{\cB}{\mathcal{B}}
\newcommand{\cC}{\mathcal{C}}

\newcommand{\cE}{\mathcal{E}}

\newcommand{\cH}{\mathcal{H}}

\newcommand{\cK}{\mathcal{K}}

\newcommand{\cU}{\mathcal{U}}

\begin{document}

%\preprint{APS/123-QED}

\title{Quantum Process Fidelity Bounds from Sets of Input States}% Force line breaks with \\

\author{Karl Mayer}
\email{karl.mayer@colorado.edu}
 \affiliation{National Institute of Standards and Technology, Boulder, Colorado, USA}
\affiliation{Department of Physics, University of Colorado, Boulder, Colorado, USA}

\author{Emanuel Knill}
\affiliation{National Institute of Standards and Technology, Boulder, Colorado, USA}
\affiliation{Center for Theory of Quantum Matter, University of Colorado, Boulder, Colorado, USA}

\date{\today}% It is always \today, today,
             %  but any date may be explicitly specified

\begin{abstract}
We investigate the problem of bounding the quantum process fidelity given bounds on the fidelities between target states and the action of a process on a set of pure input states. We formulate the problem as a semidefinite program and prove convexity of the minimum process fidelity as a function of the errors on the output states. We characterize the conditions required to uniquely determine a process in the case of no errors, and derive a lower bound on its fidelity in the limit of small errors for any set of input states satisfying these conditions. We then consider sets of input states whose one-dimensional projectors form a symmetric positive operator-valued measure (POVM). We prove that for such sets the minimum fidelity is bounded by a linear function of the average output state error. The minimal non-orthogonal symmetric POVM contains $d+1$ states, where $d$ is the Hilbert space dimension. Our bounds applied to these states provide an efficient method for estimating the process fidelity without the use of full process tomography.
\end{abstract}

\maketitle

%\tableofcontents

\section{Introduction}
\label{sec:level1}

As the complexity of small scale quantum devices continues to
increase, efficient methods for characterizing the performance of such
devices will become ever more important. A common problem is to
determine how well a process implemented by these devices matches a
unitary target process. A general tool for solving this problem is
process tomography~\cite{Chuang1997}. In a $d$-dimensional Hilbert
space, full process tomography requires preparing $d^2$ states,
applying the process to each and characterizing the final states with
informationally complete measurements. In systems with many qubits,
the resources required for full process tomography make it
prohibitively expensive. In practice, however, one is often concerned
only with estimating the process fidelity with respect to the target
process. These estimates can take the form of rigorous upper and lower
bounds, which raises the question of the resources required for
obtaining such bounds.

A method for bounding the process fidelity due to
Hofmann involves the use of two mutually unbiased
bases~\cite{Hofmann2005}. For each basis, one applies the process to
the states corresponding to the basis elements and computes the
average of the fidelities between the resulting output and the desired
target states.  These averages $F_i$, $i=1,2,$ determine bounds on the
process fidelity, where $F_{1}=F_{2}=1$ only for the target
process. This method requires $2d$ input states and measurements, a
substantial reduction in resources compared to process tomography. The
reduction comes at the cost of a gap between the lower and upper
bounds on conventional fidelities, which suggests the problems of
characterizing the tradeoff between number of input states and the gap
and of determining the minimum number of input states that are
sufficient for identifying the process.

In Ref.~\cite{Reich2013}, conditions required for the action on a set
of input states to uniquely determine a unitary process were obtained,
and a set of $d+1$ pure states satisfying the conditions was
introduced. The set contains an orthonormal basis plus a state that is
an equal superposition of the basis elements. The authors numerically
compared the process fidelity to a heuristically chosen average
between arithmetic and geometric means of the state fidelities,
finding a good correspondence between the two quantities. An exact
lower bound on the process fidelity in terms of the output state
fidelities for this set of input states in the two-qubit case was
subsequently given in Ref.~\cite{Fiurasek2014}. Such analytic
expressions for the minimum process fidelity are difficult to find in
general, with only a few examples currently
known~\cite{Micuda2013,Sedlak2016}.

In this paper, we develop a general approach for bounding the process
fidelity of a quantum process $\mathcal{E}$ with respect to a unitary
target given the fidelities for pure input states $\ket{\psi_k}$. We
first formulate the problem as a semidefinite
program~\cite{Audenaert2002}, which can be solved numerically for
any set of input states.
We then consider the case where the process acts
perfectly, that is, without error, on each input state. We give
necessary and sufficient conditions that the input states must satisfy
in order to uniquely determine the process given that the process has
unit fidelity for the input states, and show that the minimum number
of required states is $d$. In the case of errors, we derive a bound on
the process infidelity that is $O(\sqrt{\epsilon})$ in the errors. The
bound is expressed in terms of a weighted graph constructed from the
inner products of pairs of input states. Although this bound holds for
any set of input states satisfying the aforementioned conditions, it
is not tight, and we compare it with numerical solutions for random
sets of input states. Finally, we prove simple bounds on the process
fidelity for particular sets of input states, namely $N$ pure states
with $d+1\le N\le d^2$ whose projectors form a symmetric POVM. For the
minimal such set of input states, the bounds we obtain improve upon
the work of Ref.~\cite{Reich2013} and provide an efficient protocol
for bounding the process fidelity, which we compare to the method of
Ref.~\cite{Hofmann2005} for various error channels.

\section{Preliminaries}\label{Preliminaries}

Let $\mathcal{H}=\mathbb{C}^d$ denote a $d$-dimensional Hilbert space,
and $\mathcal{B}(\mathcal{H})$ the space of linear operators on
$\mathcal{H}$. For a pure state $\ket{\psi}$, we abbreviate
  $\ketbra{\psi}$ by $\hat\psi$. The identity operator is denoted
  by $I$. A \textit{quantum process} or \textit{channel} is a linear map
$\mathcal{E}:\mathcal{B}(\mathcal{H})\to\mathcal{B}(\mathcal{H})$ that
is completely positive and trace preserving
(CPTP)~\cite{Nielsen2010}. According to the Choi-Jamiolkowski
isomorphism~\cite{Choi1975,Jamiolkowski1972}, a CPTP map $\cE$ may be
represented by a density operator $\chi$ on the tensor product space
$\cH\otimes\cH$, which is defined as follows. Let $\{\ket{x}\}$ be an
orthonormal basis for $\cH$ and let
$\ket{\phi}=\frac{1}{\sqrt{d}}\sum_{x=0}^{d-1}\ket{x}\ket{x}$ be a
maximally entangled bipartite state. Then the \textit{Choi operator}
is given by
\begin{equation*}
\chi=(I\otimes\mathcal{E})(\hat\phi).
\end{equation*}
The complete positivity and trace preserving properties of $\cE$ result in the requirements that $\chi\ge0$, and that the partial trace satisfies $\mathrm{Tr}_2(\chi)=I/d$, respectively. In terms of the Choi operator, the output of the process on an arbitrary state $\rho\in\cB(\cH)$ is given by
\begin{equation}\label{process output}
    \mathcal{E}(\rho) = d\,\mathrm{Tr}_1(\chi(\rho^{\intercal}\otimes I)),
\end{equation}
where the superscript $\intercal$ on $\rho^{\intercal}$ denotes
transposition with respect to the basis $\{\ket{x}\}$. We also
  need the useful property of $\ket{\phi}$ that
\begin{equation}\label{max ent prop}
(A\otimes I)\ket{\phi}=(I\otimes A^{\intercal})\ket{\phi},
\end{equation}
for any operator $A$.

One measure of how close a process $\mathcal{E}$ comes to implementing a desired unitary operation $U$ is the \textit{average fidelity}, defined as
\begin{equation*}
    F_{\mathrm{avg}}(\mathcal{E},U)=\int d\psi\bra{\psi}U^{\dag}\mathcal{E}(\hat{\psi})U\ket{\psi},
\end{equation*}
where the integral is taken over all pure states with respect to the Haar measure. A closely 
related quantity is the \textit{entanglement fidelity}, which we simply call the process fidelity. It is defined as
\begin{equation}\label{def fidelity}
    F(\mathcal{E},U)=\bra{\phi}(I\otimes U^{\dag})\chi (I\otimes U)\ket{\phi}=\Trace{(\chi\chi_U)},
\end{equation}
where $\chi_U$ is the Choi operator for the unitary $U$. The process fidelity measures not only how well quantum information in a system is preserved, but also how well the entanglement with other systems is
preserved. The average fidelity is linearly related to
the process fidelity by the formula \cite{Nielsen2002}
\begin{equation*}
    F_\mathrm{avg}=\frac{d\,F+1}{d+1}.
\end{equation*}
For the remainder of this paper,
fidelities of processes will be taken with respect to the identity:
$F(\cE)\equiv
F(\cE,I)$. This is done without loss of generality by replacing $\cE$ with $\cU^{\dagger}\circ\cE$, where $\cU^{\dag}(\rho)=U^{\dag}\rho\,U$.

\section{Statement of problem}\label{Statement of problem}

Let
$\{\ket{\psi_k}\}_{k=1}^N$ be an indexed family of pure states in
$\mathcal{H}$, fix $\bm{\epsilon}=(\epsilon_1,\ldots,\epsilon_N)$ with $\epsilon_k\geq 0$, and let $\cC(\bm{\epsilon})$ be the
convex set of CPTP maps $\mathcal{E}$ such that for all $k$,
\begin{equation}\label{constraint}
    \bra{\psi_k}\mathcal{E}(\hat{\psi}_k)\ket{\psi_k} \ge 1-\epsilon_k.
\end{equation}
We refer to $\{\ket{\psi_{k}}\}$ as the set of input states.
We wish to find
\begin{equation*}
F_\mathrm{min}(\bm{\epsilon})=\min_{\mathcal{E}\in\mathcal{C}(\bm{\epsilon})}F(\mathcal{E}).
\end{equation*}
Note that the minimum is achieved by compactness of the feasible set.
The $\epsilon_k$ are upper bounds on the state infidelities, and can
be interpreted as errors which have been determined experimentally. If
the input states can be prepared with high fidelity, then the
$\epsilon_k$ can be obtained by measuring survival probabilities after
applying the inverse of the state-preparation transformation once
$\mathcal{E}$ has acted on the input states. If the target process is
some unitary other than the identity, often the target output states
can be mapped back to the measurement basis by high-fidelity
unitaries, in which case the state fidelities
$\bra{\psi_k}\mathcal{E}(\hat{\psi}_k)\ket{\psi_k}$ can be obtained
directly. Otherwise the state fidelities can be obtained via direct
fidelity estimation~\cite{Flammia2011}, which for qubit systems
requires only one-qubit gates and Pauli basis
measurements, and a number of experimental trials that grows linearly
in $d$.

We also consider the situation where upper bounds on the state fidelities are known. In this case, the problem is to find
\begin{equation*}
F_\mathrm{max}(\bm{\epsilon})=\max_{\mathcal{E}\in\cK(\bm{\epsilon})}F(\cE),
\end{equation*}
where $\cK(\bm{\epsilon})$ is the convex set of CPTP maps $\cE$
that satisfy
\begin{equation*}
    \bra{\psi_k}\mathcal{E}(\hat{\psi}_k)\ket{\psi_k} \le 1-\epsilon_k
\end{equation*}
for all $k$.

The bounds $F_\mathrm{min}$ or $F_\mathrm{max}$ can be found numerically by solving a semidefinite program (SDP)~\cite{Vandenberghe1996}. To formulate our problem as an SDP, we use the Choi matrix representation and Eq.~\ref{process output}.
Our task is then to solve
\begin{align}\label{SDP}
\begin{array}[b]{ll}
\textrm{Minimize:}&\mathrm{Tr}(\chi\hat\phi)\\
\textrm{Subject to:}&\mathrm{Tr}_2(\chi)=I/d,\\
&\mathrm{Tr}(\chi\,(\hat{\psi}_k^{\intercal}\otimes\hat{\psi}_k))\ge\frac{1}{d}(1-\epsilon_k),\\
&\chi\ge0.
\end{array}
\end{align}
A number of software packages are available for efficiently solving
SDP's; for this work we used cvx~\cite{Boyd2004,CVXResearch2012}.  We
thus have a numerical solution to our posed problem: once the
experimenter has determined the $\epsilon_k$, they can then solve the
above SDP to obtain $F_\mathrm{min}$ as a lower bound for the process
fidelity. However, the experimenter may wish to know which set of
input states to prepare in order to get a good lower bound. We
  therefore investigate properties of the solution to the SDP given
by Eq.~\ref{SDP}, both in general and for special cases with
particular errors or input states.

\section{Convexity}

Our first observation is that the minimum process fidelity is a convex function of the error bounds $\epsilon_k$.

\begin{proposition}
  $F_\mathrm{min}(\bm{\epsilon})$ is convex, that is,
  $F_\mathrm{min}(t\bm{\epsilon}_1+(1-t)\bm{\epsilon}_2)\le t
  F_\mathrm{min}(\bm{\epsilon}_1)+(1-t)F_\mathrm{min}(\bm{\epsilon}_2)$
  for all $t\in[0,1]$.
\end{proposition}

\begin{proof}
Let $\mathcal{E}_1\in\mathcal{C}(\bm{\epsilon}_1)$ and $\mathcal{E}_2\in\mathcal{C}(\bm{\epsilon}_2)$
satisfy $F(\mathcal{E}_1)=F_\mathrm{min}(\bm{\epsilon}_1)$ and $F(\mathcal{E}_2)=F_\mathrm{min}(\bm{\epsilon}_2)$, and consider $\mathcal{E}=t\,\mathcal{E}_1+(1-t)\,\mathcal{E}_2$. By linearity of the process fidelity, 
\begin{equation*}
    F(\mathcal{E})=t F_\mathrm{min}(\bm{\epsilon}_1)+(1-t)F_\mathrm{min}(\bm{\epsilon}_2).
\end{equation*}
But for all $k$,
\begin{align*}
    \bra{\psi_k}\mathcal{E}(\hat{\psi}_k)\ket{\psi_k} &=\bra{\psi_k}(t\,\cE_1(\hat{\psi}_k)+(1-t)\,\cE_2(\hat{\psi}_k))\ket{\psi_k}\\
    &\ge t\,(1-(\bm{\epsilon}_1)_k)\;+\;(1-t)\,(1-(\bm{\epsilon}_2)_k)\\
    &= 1 -t(\bm{\epsilon}_1)_k-(1-t)(\bm{\epsilon}_2)_k,
\end{align*}
where the second line follows from
$\mathcal{E}_{i}\in\mathcal{C}(\bm{\epsilon}_{i})$.  Consequently
$\mathcal{E}\in\mathcal{C}(t\bm{\epsilon}_1+(1-t)\bm{\epsilon}_2)$,
and therefore
\begin{align*}
  F_\mathrm{min}(t\bm{\epsilon}_1+(1-t)\bm{\epsilon}_2)&\le F(\mathcal{E})\\
  &= t F_\mathrm{min}(\bm{\epsilon}_1)+(1-t)F_\mathrm{min}(\bm{\epsilon}_2).
\end{align*}
\end{proof}
The convexity of the minimum process fidelity is illustrated in
Fig.~\ref{fig:fidelity vs error}, which shows a plot for
$\epsilon_k=\epsilon$ of $F_\mathrm{min}(\epsilon)$
versus $\epsilon$ for two sets of input states to be defined in
Sec.~\ref{SINC POVM} (we use an unbold $\epsilon$ in
$F_\mathrm{min}(\epsilon)$ to denote constant $\epsilon_k$).  A useful
consequence of the convexity property is that a lower bound on the
process fidelity can be obtained from a tangent line of
  $F_{\mathrm{min}}(\epsilon)$ at $\epsilon=0$.

For the function $F_{\mathrm{max}}$ we have:
\begin{proposition}
$F_\mathrm{max}(\bm{\epsilon})$ is concave, that is, $F_\mathrm{max}(t\bm{\epsilon}_1+(1-t)\bm{\epsilon}_2)\ge t F_\mathrm{max}(\bm{\epsilon}_1)+(1-t)F_\mathrm{max}(\bm{\epsilon}_2)\;\forall t\in[0,1].$
\end{proposition}
The proof can be obtained by following the proof of convexity of
  $F_{\mathrm{min}}$, replacing $\min$ by $\max$ and reversing
  inequalities as necessary.

\section{Processes with $\epsilon=0$}
\label{sec:level2}

  In this section we analyze the special case where each $\epsilon_k=0$.
  If the only process $\mathcal{E}$ with
  $\mathcal{E}(\hat\psi_{k})=\hat\psi_{k}$ for all $k$ is the identity
  process, we say that the set of input states identifies unitaries.
  Identifying unitaries is equivalent to $F_{\mathrm{min}}(0)=1$.  The
  next theorem characterizes sets of input states that identify
  unitaries.  Define the graph $G=(V,E)$ by assigning vertex set
  $V=\{k\}$ and edge set $E=\{(k,k^{\prime}):
  \bra{\psi_{k}}\ket{\psi_{k^{\prime}}}\ne 0\}$.
\begin{theorem}\label{theorem1}
  The set of input states identifies unitaries iff the input states span $\mathcal{H}$ and the graph $G$ is
    connected.
\end{theorem}

\begin{proof}
  Suppose that the input states span $\mathcal{H}$ and the graph
    $G$ is connected. By dilation, any CPTP map can be
  expressed in the form
  \begin{equation*}
    \mathcal{E}(\rho)=\mathrm{Tr}_2(U(\rho\otimes\ket{0}\bra{0})U^{\dag}),
  \end{equation*}
  for some ancillary state $\ket{0}$ and unitary $U$ on the joint
  input-ancilla system. Suppose that
    $\mathcal{E}(\hat\psi_{k})=\hat\psi_{k}$ for all $k$. Since
  $\Tr_{2}(U (\hat \psi_{k}\otimes \hat 0) U^{\dagger}) =
  \mathcal{E}(\hat\psi_{k}) = \hat\psi_k$ is pure,
  $U\ket{\psi_k}\ket{0}$ is a product state:
  $U\ket{\psi_k}\ket{0}=\ket{\psi_k}\ket{e_k}$, where $\ket{e_k}$ is
  an ancilla state which may depend on $k$.
  We prove that $\ket{e_{k}}$ is independent of $k$.  From the identity
  \begin{equation*}
    \bra{\psi_{k^{\prime}}}\ket{\psi_k}=
    \bra{0}\bra{\psi_{k^{\prime}}}U^{\dag}U\ket{\psi_k}\ket{0}=
    \bra{\psi_{k^{\prime}}}\ket{\psi_k}\bra{e_{k^{\prime}}}\ket{e_k}, 
  \end{equation*}
  it follows that if $k$ and $k^{\prime}$ are
  adjacent in $G$ then $\bra{e_{k^{\prime}}}\ket{e_k}=1$.
  Since $G$ is connected, we conclude that all
  the $\ket{e_k}$ are equal, and with $\ket{e}=\ket{e_{k}}$, we
    have $U\ket{\psi_k}\ket{0}=\ket{\psi_k}\ket{e}$ for all $k$. By
  linearity of $U$ and since the $\ket{\psi_k}$ span $\mathcal{H}$, it
  follows that $U\ket{\psi}\ket{0}=\ket{\psi}\ket{e}$ and
  $\mathcal{E}(\hat\psi)=\hat\psi$ for all pure states
  $\ket{\psi}$. By linearity of quantum processes,
  $\mathcal{E}(\rho)=\rho$ for all density matrices
  $\rho$, and $\mathcal{E}$ is the identity process.

   For the reverse implication, we prove the contrapositive.
    Suppose first that the input states do not span $\cH$. Let $S$ be
    the span of $\{\ket{\psi_k}\}$, and $S^{\perp}$ the orthogonal
    complement of $S$. Then $U=I_S\oplus i I_{S^{\perp}}$ has fidelity
    $1$ on all input states, but is not the identity process.  Next,
    suppose that $G$ is disconnected. Let $S$ be the span of the
    states $\ket{\psi_k}$ in a connected component of $G$.  Then
    $S\not=0$ and $S^{\perp}\not=0$, and again $U=I_S\oplus i
    I_{S^{\perp}}$ has fidelity one on the input states but is not the
    identity process.
\end{proof}

  Sets of input states that identify unitaries are also characterized
  by having trivial commutant, meaning $K=\{U\in SU(d): 
  [U,\hat{\psi}_k]=0\textrm{ for all } k\}=\{I\}$. Indeed, we show in the 
  appendix, Prop.~\ref{prop:uic=gc},
  that $K=\{I\}$ iff the input states are spanning and $G$ is connected.
  Our characterization is related to an observation made in Ref.~\cite{Reich2013}: if a set of
 states $\{\rho_k\}_k$ has trivial commutant, then every unitary $U$ is
 uniquely determined by its action on the states $U\rho_k U^{\dag}$. A set 
 of states with this property is called unitarily informationally complete 
 (UIC) \cite{Baldwin2014}. For pure input states and unitary processes, 
 the UIC property is equivalent to the property that if
$\mathcal{E}(\hat\psi_k)=\hat\psi_k$ for all $k$, then $U=I$.
Our Thm.~\ref{theorem1} together with the
  mentioned Prop.~\ref{prop:uic=gc} is therefore a strengthening of
  the observation from Ref.~\cite{Reich2013} above. In particular,
 for \textit{any}
 process $\mathcal{E}$, not just unitary processes, if the input states have trivial
  commutant, then having $\mathcal{E}(\hat{\psi}_k)=\hat{\psi}_k$ for all $k$ is
sufficient for $\mathcal{E}=I$. We remark
  that compared to checking for a trivial commutant, it is simpler to
  check the properties that the input states are spanning and $G$ is
  connected.
  
The authors of Ref.~\cite{Reich2013} also provided an example of a set
of $d+1$ pure states with the UIC property. This set contains the $d$
computational basis states $\ket{0},...,\ket{d-1}$, as well as the
``totally rotated state", defined as
$\ket{\psi_{tr}}=\frac{1}{\sqrt{d}}\sum_x\ket{x}$. The authors claimed
that this set contains the minimum number of pure
states required to uniquely determine a unitary process. However,
Thm.~\ref{theorem1} implies that $d$ states suffice. The
simplest example has $d=2$ and consists of any
two non-orthogonal pure states.

\section{Fidelity Lower Bounds}\label{sec: errors}

We have shown that the minimum number of input
states sufficient to ensure that the process fidelity equals unity in the limit
of no errors is equal to the dimension $d$. We now consider the case
of small non-zero errors $\epsilon_k$. Suppose that the input states are spanning and
  $G$ is connected.  We obtain a lower bound for
$F_{\mathrm{min}}(\bm{\epsilon})$ to lowest order in the $\epsilon_k$.
To describe the lower bound, order the input states so
that the first $d$ input states are spanning.
Let $M$ be the Gram matrix for the states $\{\ket{\psi_k}\}_{k=1}^d$, defined as the $d$-by-$d$
matrix with entries $M_{kk'}=\braket{\psi_k}{\psi_{k'}}$.
For the lower bound, we also need to introduce a minimum-weight
path quantity $W_{kk'}$ defined as follows:
Let $\mathcal{P}_{k k^{\prime}}$ denote the set of paths in $G$ from vertex $k$ to $k^{\prime}$. Then
\begin{equation}\label{weight}
    W_{kk'}=\min_{P\in\mathcal{P}_{kk'}}\sum_{(i,j)\in P}\abs{M_{ij}}^{-1/2}.
\end{equation}
With these definitions we can
  establish the following:
\begin{theorem}\label{thm:fid_lowerbnd}
  Let $\epsilon=\max_{k}\epsilon_{k}$. For all $\cE\in\cC(\bm{\epsilon})$,
  \begin{align*}
    F(\cE)&\ge 1 - C\sqrt{\epsilon} + O(\epsilon), \\
    C &= \frac{2}{d}\Big(\frac{2}{d}\sum_{d\geq k>k'\geq 1}W_{kk'}^2+\sum_{k=1}^d\sqrt{(M^{-1})_{kk}}\Big).\notag\\
    &\hphantom{\ge\;}
  \end{align*}
\end{theorem}
  The proof of the theorem is in the appendix
  Sec.~\ref{ap:fid_lowerbnd}, where it is established by proceeding
  along the same lines as the proof of Thm.~\ref{theorem1} while
  explicitly keeping track of error terms to lowest order. A
  refinement of the bound taking into account non-constant
  $\epsilon_{k}$ is described at the end of the proof.

The quantity $W_{kk'}$ can be found in $O(N^2)$ time
with algorithms for minimum weighted
paths~\cite{Dijkstra1959}. Note that $W_{kk'}$ is large if two
adjacent states on the minimal path are nearly orthogonal. The matrix
$M$ is invertible if, as we assume, the states $\ket{\psi_k}$ span
$\cH$, and the diagonal entries of $M^{-1}$ are large if any two
states are nearly equal. The lower bound given by Thm.~\ref{thm:fid_lowerbnd} can thus be understood as quantitatively enforcing
the conditions of Thm.~\ref{theorem1}.

A few comments are in order. First, note that the lowest order term in
the lower bound is of order $\sqrt{\epsilon}$. This scaling behavior
matches our empirical observations from numerically solving the SDP
given by Eq.~\ref{SDP}. However, we find that for
many sets of input states containing more
than $d$ states, the $O(\sqrt{\epsilon})$ term vanishes and the
process infidelity becomes linear in $\epsilon$ for small
error. Examples include the basis states plus the totally rotated
state, as well as the symmetric POVM states defined in the next
section. The transition from sub-linear to linear scaling is not
explained by the proof of Thm.~\ref{thm:fid_lowerbnd} and thus
remains an open question.  Second, the bound in Thm.~\ref{thm:fid_lowerbnd} is not tight. Fig.~\ref{fig:rootEpsRand}
compares the upper bound for the $O(\sqrt{\epsilon})$ term with its true value obtained via SDP, for 100 random sets of
$N=d$ input states of dimensions $d=4,8$. The $O(\sqrt{\epsilon})$ terms were obtained by computing $F_{\min}(\epsilon)$ for $\epsilon$ varying between $10^{-5}$ to $10^{-4}$ in steps of $10^{-5}$, and performing a linear least squares best fit. The plot shows that the
bound tends to overestimate the process infidelity by a factor of approximately two
  for these dimensions, and that the fractional discrepancy decreases
  as the $O(\sqrt{\epsilon})$ term increases. Improving the lower bound 
  Thm.~\ref{thm:fid_lowerbnd} remains an open problem.
  
\begin{figure}
    \centering
    \includegraphics[scale=.50]{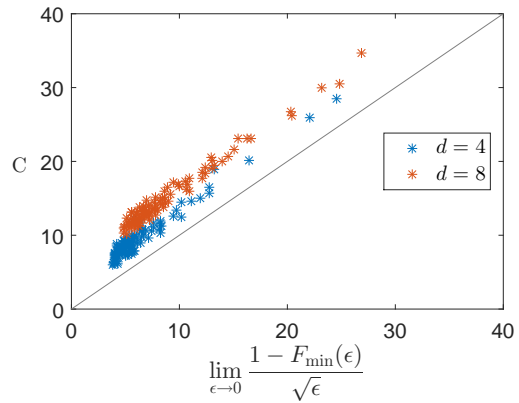}
    \caption{Scatter plot of the
      $\sqrt{\epsilon}$ coefficient in the expansion of
      $F_{\min}(\epsilon)$ inferred via SDP against the quantity $C$ in 
      Thm.~\ref{thm:fid_lowerbnd}, which is an upper bound on the $\sqrt{\epsilon}$ coefficient. Each data point corresponds to
      a set of $d$ Haar random input states. 
      A gray line with slope $m=1$ is included to aid the comparison.}
    \label{fig:rootEpsRand}
\end{figure}

\section{Symmetric POVM input states}\label{SINC POVM}

In this section, we prove lower and upper bounds on the process
fidelity for a set of $N$ input states whose one-dimensional
projectors form a symmetric positive operator-valued measure
(POVM). Such sets are also referred to as equiangular tight
frames \cite{Sustik2007}. We show that for such sets of input states,
$F_\mathrm{min}$ is bounded by a linear function of the errors $\epsilon_k$. Our
motivation for studying
symmetric POVM input states is that they are in a sense
maximally spread out in the Hilbert space $\mathcal{H}$, and are
therefore good candidates for yielding
the tightest possible bounds for a given $N$.

The set of input states forms a symmetric POVM if its states have constant pairwise overlap and the sum of the
input projectors is proportional to the identity. That is, the input states satisfy that for some
  constant $c\geq 0$ and for all $k\not=k'$
\begin{align}
  \abs{\bra{\psi_k}\ket{\psi_{k^{\prime}}}}^2 &= c,
     \label{state overlap}\\
    \frac{d}{N}\sum_k \ket{\psi_k}\bra{\psi_k} &= I,  \label{POVM}
\end{align}
where the factor $\frac{d}{N}$ is determined by matching the traces of
the two sides of the identity.  
By squaring Eq.~\ref{POVM} and taking the trace, the
constant in Eq.~\ref{state overlap} is found to be
\begin{equation}\label{constant c}
c = \frac{N-d}{d(N-1)}.
\end{equation}

  Conventionally, a POVM consists of a family of positive
  semidefinite hermitian operators summing to the identity. We
  slightly abused the terminology in referring to the set of input
  states as a POVM. The conventional POVM formed from the input states
  satisfying Eqs.~\ref{state overlap} and~\ref{POVM} consists of the
  operators $\frac{d}{N}\hat \psi_{k}$. If the set of input states forms a
  symmetric POVM, then the input states are spanning. If $N>d$, the
graph $G$ defined at the beginning
  of Sec.~\ref{sec:level2} is a complete graph.  When $N=d^2$, the
input projectors $\hat\psi_k$ span $B(\cH)$
and therefore form a symmetric informationally complete (SIC) POVM
\cite{Renes2004}. At the other extreme, the smallest non-trivial
symmetric POVM occurs when $N=d+1$, because for $N=d$ we
have $c=0$ and $G$ is not connected. A set of $d+1$ states
satisfying Eqs.~\ref{state overlap} and \ref{POVM} is called a
\textit{simplex}.

Whereas SIC POVMs are conjectured but not proven to exist in all
dimensions~\cite{Fuchs2017}, we give an
explicit construction of a simplex.  Let $\omega = \exp(2\pi i/(d+1))$
be a $(d+1)^{\mathrm{th}}$ root of unity. For $k\in\{0,1,..,d\}$,
define
\begin{equation}\label{def psi k}
    \ket{\psi_k}=\frac{1}{\sqrt{d}}\sum_{x=0}^{d-1}\omega^{kx}\ket{x}.
\end{equation}
By direct calculation one can confirm that
Eqs.~\ref{state overlap} and~\ref{POVM} are satisfied.

Symmetric POVM input states have the nice
property that $F_{\min}(\epsilon)$ is linear for constant $\epsilon_k=\epsilon$. 
An example is shown in Fig.~\ref{fig:fidelity vs error}, which shows
  $F_{\min}(\epsilon)$ when the set of input states is a simplex.
  The figure also shows
  $F_{\min}(\epsilon)$ for the set of input states
$\{\ket{0},...,\ket{d-1},\ket{\psi_{tr}}\}$ from
Ref.~\cite{Reich2013}, which is not linear and has a more negative
  slope as $\epsilon$ goes to zero. This demonstrates
that the simplex is a better choice of $d+1$ input states for
obtaining lower bounds on the process fidelity. We
conjecture that symmetric POVM states are optimal among
all sets of $N$ input states in this regard.

\begin{figure}
    \centering
    \includegraphics[scale=.50]{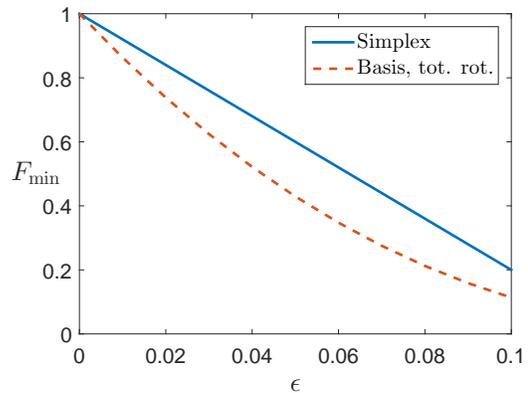}
    \caption{Plot of $F_{\mathrm{min}}(\epsilon)$ for two sets of input states in $d=8$ dimension. The solid curve is for the simplex, the dashed for
 the standard basis together with the totally rotated state defined in the text.
 Note that $F_{\mathrm{min}}$ is convex.} 
    \label{fig:fidelity vs error}
\end{figure}

Our main result on the performance of symmetric POVMs is a general,
  linear bound on $F(\cE)$. After the proof we show
  that the lower bound is tight for $\epsilon_{k}=\epsilon$ constant.
\begin{theorem}\label{theorem2}
    Suppose that the set of input states with
    $N>d$ forms a symmetric POVM and let
  $\mathcal{E}$ be a CPTP map such that
  $1-u_k\ge\bra{\psi_k}\mathcal{E}(\hat\psi_k)\ket{\psi_k}\ge
  1-\epsilon_k$ for all $k$. Then 
  \begin{equation*}
      1-\bar{u}\ge F(\mathcal{E})\ge 1 -
  \bigg(\frac{N-1}{N-d}\bigg)\,\bar{\epsilon},
  \end{equation*}
  where $\bar{\epsilon}$ and
  $\bar{u}$ are the means of the $\epsilon_k$ and $u_k$.
\end{theorem}

\begin{proof}
  We first prove that $F(\cE)\geq
  1-(\frac{N-1}{N-d})\,\bar{\epsilon}=1-\frac{1}{dc}\bar{\epsilon}$, where
    $c$ is defined in Eq.~\ref{constant c}. We apply the
    assumed bounds and Eq.~\ref{process output} to obtain
  \begin{align*}
    (1-\epsilon_{k})&\le \bra{\psi_{k}}\cE(\hat\psi_{k})\ket{\psi_{k}}\notag\\
    &= d\, 
    \Tr(\hat\psi_{k}\Tr_{1}\left(\chi\,\left(\hat\psi_{k}^{\intercal}\otimes
      I\right)\right))\notag\\
    &= d\,\Tr(\chi\,\left(\hat\psi_k^{\intercal}\otimes\hat\psi_k\right)).
  \end{align*}
  Define $\Gamma=\sum_k\hat\psi_k^{\intercal}\otimes\hat\psi_k$.
   Summing the inequality just obtained over $k$ and dividing by $N$ gives
    $1-\bar{\epsilon} \le \frac{d}{N} \mathrm{Tr}(\chi \Gamma)$, which
    is equivalent to
  \begin{equation}
    1-\frac{1}{dc}\bar{\epsilon}\le
    1-\frac{1}{dc} + \frac{1}{Nc}\mathrm{Tr}(\chi \Gamma).
    \label{eq:lowerbnd1}
  \end{equation}
    The inequality to
    be proven follows once we show that $F(\cE)$ is bounded below by
    the right-hand-side. Since $F(\mathcal{E})=\mathrm{Tr}(\chi\hat\phi)$,
    this is implied by
  \begin{equation*}
    \mathrm{Tr}(\chi\hat\phi)\ge 1-\frac{1}{dc} + 
    \frac{1}{Nc}\mathrm{Tr}(\chi \Gamma).
  \end{equation*}
  After moving everything to the left-hand-side and defining
    \begin{equation*}
      A = \hat\phi+\left(\frac{1}{dc}-1\right)\,I-\frac{1}{Nc}\Gamma,
      %\label{eq:lowerbnd2}
    \end{equation*}
    we can see that the desired inequality is equivalent to $\mathrm{Tr}(\chi A)\ge0$,
    and it suffices to prove that $A$ is
    positive semidefinite. For this purpose, we determine the spectral
    decomposition of $\Gamma$. We can write $\Gamma=BB^{\dag}$ with $B$ given by
  \begin{equation*}
    B = \sum_{k=1}^{N}\left(\ket{\psi_k^{*}}\otimes\ket{\psi_k}\right)\bra{k},
  \end{equation*}
    where $\ket{\psi_{k}^{*}}$ is the complex conjugate of
    $\ket{\psi_{k}}$ relative to the standard basis, and we used
    $\hat\psi_{k}^{*}=\hat\psi_{k}^{\intercal}$. $B$ is a matrix of dimension $d^{2}\times N$ and 
    the spectrum of $BB^{\dag}$ is the same as that of $B^{\dag}B$, which
    is the $N\times N$ matrix whose $k,k'$ entry is given by
    \begin{equation*}
      \abs{\bra{\psi_{k}}\ket{\psi_{k'}}}^{2}
      =((1-c)\delta_{kk'} + c).
    \end{equation*}
    With respect to the basis consisting of the
    $\ket{k}$, this is a matrix whose diagonal entries are
    $1$ and whose off-diagonal entries are $c$. Such a matrix has two
    eigenvalues: the first is $(N-1)c+1 = \frac{N}{d}$ corresponding
    to the eigenvector with constant entries, and the second is
    $(1-c)$ with multiplicity $N-1$.  Accordingly, we can write
    \begin{equation}
      \Gamma=BB^{\dag}=\frac{N}{d}\hat\phi'+(1-c)\Pi,\label{eq:lowerbnd3}
    \end{equation}
    where $\hat\phi'$ is a rank-one projector and $\Pi$ is a
    rank $N-1$ projector orthogonal to $\hat\phi'$.
    We determine that $\hat\phi'=\hat\phi$ by verifying that
  $\ket{\phi}$ is an eigenstate of $\Gamma$ with eigenvalue $N/d$:
  From Eqs.~\ref{max ent prop} and~\ref{POVM},
  \begin{align*}
    \Gamma\ket{\phi} &=   
    \sum_k(\hat\psi_k^{\intercal}\otimes\hat\psi_k)\ket{\phi}\\
    &=\sum_k(I\otimes\hat\psi_k)\ket{\phi}\\
    &=\frac{N}{d}\ket{\phi}.
  \end{align*}
  Let $\Pi_{\perp}=I-\hat\phi-\Pi$ be the projector
    onto the nullspace of $\Gamma$. We can now write $A$ as
    \begin{align*}
      A &=
      \hat\phi+\left(\frac{1}{dc}-1\right)(\hat\phi+\Pi+\Pi_{\perp})\\
      &\hphantom{=\;}-\left(\frac{1}{dc}\hat\phi+\frac{1-c}{Nc}\,\Pi\right)\\
      &= \frac{d-1}{N-d}\,\Pi_{\perp},
    \end{align*}
    since $\frac{1}{dc}-1 =\frac{1-c}{Nc} = \frac{d-1}{N-d}$.
    Thus $A$ is positive semidefinite as claimed.

  The proof that $F(\cE)\leq 1-\bar u$ follows the same strategy.
    Instead of Eq.~\ref{eq:lowerbnd1}, from the upper bound on the
    input state fidelities we derive $1-\bar{u}\ge
    \frac{d}{N}\mathrm{Tr}(\chi \Gamma)$.  The inequality to be proven
    now is implied by $\mathrm{Tr}(\chi\hat\phi)\le
    \frac{d}{N}\mathrm{Tr}(\chi \Gamma)$, so it suffices to show that
    $\Tr(\chi A')\leq 0$ with
    $A'=\hat\phi-\frac{d}{N}\Gamma$. Applying Eq.~\ref{eq:lowerbnd3}
    gives $A' = -\frac{d}{N}(1-c)\Pi$, which is negative semidefinite
    since $c<1$.
  \end{proof}

Numerical solutions indicate that the lower bound
$F(\cE)\ge 1-(\frac{N-1}{N-d})\bar{\epsilon}$ of Thm.~\ref{theorem2} is not tight.
Determining $F_{\min}(\bm{\epsilon})$ for symmetric POVM input states and general
$\epsilon_k$ remains an open problem. However, if $\epsilon_k=\epsilon$ for all $k$,
then the lower bound is tight and achieved by the quantum channel
\begin{equation}\label{min channel}
    \mathcal{E}(\rho)=
    \left(1-p\right)\rho+p\frac{d}{N}\sum_k\hat\psi_k\rho\,\hat\psi_k,
\end{equation}
where $p=\frac{d(N-1)}{(d-1)(N-d)}\epsilon$.
 The Kraus operators for $\mathcal{E}$ are $\sqrt{1-p}\,I$ and
  $\sqrt{p\frac{d}{N}}\,\hat\psi_{k}$ for $k=1,\ldots,N$. We verify that
  $\cE$ satisfies $\bra{\psi_{k}}\cE(\hat\psi_{k})\ket{\psi_{k}}=1-\epsilon$ for
  all $k$  and $F(\mathcal{E})=1-(\frac{N-1}{N-d})\epsilon$:
  \begin{align*}
    \bra{\psi_{k}}\cE(\hat\psi_{k})\ket{\psi_{k}} &=
    (1-p)+ p\frac{d}{N}\sum_{k'}\abs{\bra{\psi_{k}}\ket{\psi_{k'}}}^{4}\\
    &= (1-p) + p\frac{d}{N}\big(1+(N-1)c^2\big)\\
%    &=1-p\left(\frac{N-d}{N}-\frac{d(N-1)}{N} c^{2}\right)\\
    &=1-p\frac{1}{N}\left(N-d-\frac{(N-d)^{2}}{d(N-1)}\right)\\
    &=1-p\frac{N-d}{dN(N-1)}\left(d(N-1)-(N-d)\right)\\
    &=1-p\frac{(N-d)(d-1)}{d(N-1)}\\
    &=1-\epsilon,
  \end{align*}
  and 
  \begin{align*}
    F(\cE) &= \bra{\phi}(I\otimes\cE)(\hat\phi)\ket{\phi}\\
    &= (1-p) + p\frac{d}{N}
    \sum_{k}\abs{\bra{\phi}(I\otimes\hat\psi_{k})\ket{\phi}}^{2}\\
    &= (1-p) + p\frac{d}{N}\sum_{k}\frac{1}{d^{2}}\\
    &= 1-p\frac{d-1}{d}\\
    &= 1-\frac{N-1}{N-d}\epsilon.
  \end{align*}
When $N=d^2$, symmetric POVM input states form a SIC POVM and therefore also a
2-design~\cite{Renes2004}. An argument similar to that in
Ref.~\cite{Nielsen2002} shows that the fidelity minimizing 
channel $\cE$ defined in Eq.~\ref{min channel} is the depolarizing channel
\begin{equation*}
\cE(\rho)=(1-q)\rho+\frac{q}{d}I,
\end{equation*}
with $q=\frac{d}{d-1}\epsilon$. For the simplex, when
$N=d+1$, the fidelity minimizing channel is in general more
difficult to interpret. For the case of $d=2$ and with the explicit simplex states
given in Eq.~\ref{def psi k},
\begin{equation*}
    \mathcal{E}(\rho)=(1-2\epsilon)\rho+\epsilon(\sigma_{x}\rho \sigma_{x}+\sigma_{y}\rho \sigma_{y}),
\end{equation*}
 where the $\sigma_{u}$ are the standard Pauli matrices.
  As can be seen, $\cE$ is a sum of
  the $x$ and $y$ dephasing channels. In a Bloch-sphere-deformation picture, the effect
  is to maximize contraction parallel to the $z$-axis while keeping
  contraction parallel to the other axes fixed. The $z$-axis
  contraction is limited by the ``no pancake theorem"
\cite{Blume-Kohout2010}, which states that there is no quantum channel
that projects the Bloch sphere onto the $x-y$ plane.

\section{Comparison to Hofmann bounds}

 Consider $N=d+1$ with $d=2^{n}$, where the state space is that of
  $n$ qubits. The set of simplex input states of Eq.~\ref{def psi k} can be
  used in an efficient experimental procedure for bounding the
  fidelity of a process. These input
states factor according to
\begin{equation*}\label{psik factorization}
\ket{\psi_k}=\bigotimes_{j=1}^n \frac{1}{\sqrt{2}}(\ket{0}+e^{i\pi k \frac{2^{n+1-j}}{d+1}}\ket{1}),
\end{equation*}
and can therefore be prepared with
one-qubit Hadamard gates and rotations about the $z$-axis. If the measured state
  fidelities satisfy
$\bra{\psi_k}\mathcal{E}(\hat{\psi}_k)\ket{\psi_k}=1-\epsilon_{k}$, 
then according to Thm.~\ref{theorem2}
the process fidelity is bounded by
\begin{equation}\label{my bounds}
  1-d\,\bar{\epsilon}\le F(\mathcal{E})\le 1-\bar{\epsilon}.
\end{equation}
We compare these bounds to those given by Hofmann~\cite{Hofmann2005}, which 
require as input states the members of
 two mutually unbiased bases (MUBs). 
A particular pair of such bases consists of
the computational basis $\{\ket{x}\}$ and
its Fourier transform $\{\ket{f_x}\}$ given by
\begin{equation*}
  \ket{f_x}=\frac{1}{\sqrt{d}}\sum_y e^{2\pi i x y/d}\ket{y}.
\end{equation*}
Thus, the Hofmann bounds require $2d$ input states, a quadratic
improvement over full process tomography in the number of states
needed to probe the fidelity of a process. 
The bounds are determined by the two classical fidelities
\begin{align*}
    F_1 &= \frac{1}{d}\sum_x\bra{x}\mathcal{E}(\hat{x})\ket{x},\\
    F_2 &= \frac{1}{d}\sum_x\bra{f_x}\mathcal{E}(\hat{f}_x)\ket{f_x},
\end{align*}
  in terms of which they are given by
\begin{equation}\label{hofmann bounds}
    F_1+F_2-1\le F(\mathcal{E})\le\mathrm{min}\{F_1,F_2\}.
\end{equation}
 Suppose that the fidelities for the input states used to apply the
  Hofmann bounds are $1-\epsilon_{k}$. Then
  $F_{1}+F_{2}-1=1-2\bar\epsilon$ and if $\epsilon_{k}=\epsilon$ is
  constant, $\mathrm{min}\{F_{1},F_{2}\}=1-\bar\epsilon$. For
  comparison, according to Thm.~\ref{theorem2}, a symmetric POVM with
$N=2d$ yields lower and upper bounds of
$1-(\frac{2d-1}{d})\bar\epsilon$ and $1-\bar\epsilon$. Assuming identical average errors,
  the lower bound is slightly tighter. The set of simplex input
  states consist of $d+1$ states, further reducing
the number of input states by a factor approaching two. Because fewer input states are used, the
bounds obtained with the simplex are looser than the
Hofmann bounds. However, the improvement obtained
from the Hofmann bounds depends on the particular process
$\mathcal{E}$. For instance, if the system is subject to an error
channel that is a depolarizing channel $\mathcal{D}(\rho) = (1-p)\rho
+ \frac{p}{d}I,$ with
$\bra{\psi_k}\mathcal{D}(\hat{\psi}_k)\ket{\psi_k}=1-\epsilon$, then for the
simplex input states one finds that $F_1=F_2=1-\epsilon$ and so
the Hofmann bounds are
\begin{equation*}
1-2\epsilon\le F(\mathcal{D})\le 1-\epsilon.
\end{equation*}
The width of the interval between the lower and upper Hofmann bounds
is smaller than that of Eq.~\ref{my bounds} by a factor of $d-1$, so
the advantage gained from using more input states grows
linearly with the dimension.
However, if the system encounters errors described by the process in Eq.~\ref{min channel}, the classical fidelities are $F_1=1-d\epsilon$ and $F_2=1-\frac{d+1}{3}\epsilon$ (see appendix), giving the bounds
\begin{equation}\label{min chan hoff bounds}
1-\Big(\frac{4d+1}{3}\Big)\epsilon\le F(\mathcal{E})\le 1-d\,\epsilon.
\end{equation}
In this case the Hofmann bounds are tighter than the bounds in
Eq.~\ref{my bounds} by a factor
  approaching three for large dimensions. Interestingly, for the process given
by Eq.~\ref{min channel}, the upper bound obtained from Eq.~\ref{hofmann bounds} and the lower bound from Eq.~\ref{my bounds} coincide. So for
this particular channel, the classical fidelities for the Hoffman input states together with the average fidelity for the simplex input states 
determine the process fidelity exactly.

\section{Conclusion}
    
  We have characterized sets of pure input states that identify
  unitary processes, and determined that the minimum number of
  states required is equal to the Hilbert space dimension $d$.
  We obtained a lower bound on $F(\cE)$ for small $\epsilon$ of the form 
  $1- C\sqrt{\epsilon}-O(\epsilon)$
  (Thm.~\ref{thm:fid_lowerbnd}). We have also proven bounds on $F(\cE)$ for symmetric POVM
  input states and shown that the lower bound is achieved for constant
  $\epsilon_k=\epsilon$. When $N=2d$, these bounds are slightly tighter than the Hofmann
  bounds obtained from a set of input states consisting of two MUBs. The smallest set of
  symmetric POVM input states which identifies unitaries is the simplex, with $N=d+1$. For
  qubit systems where $d=2^n$, simplex input states can be prepared with a circuit
  containing only Hadamard gates and individual $z$-axis rotations. However, the bounds
  obtained are in general much looser than the Hofmann bounds.

  There are a number of open problems to be investigated. 
  As noted, the bound given by
  Thm.~\ref{thm:fid_lowerbnd} is not
  tight. Is there a tight bound expressed analytically in terms of the
  input states? What property of the input states determines the
  vanishing of the $O(\sqrt{\epsilon})$ term? Another  
  open question is to find $F_{\min}(\bm{\epsilon})$ and the
  fidelity minimizing channel for symmetric POVM input states and arbitrary 
  $\bm{\epsilon}$. A general
  problem is to determine, given $N$ and $\bar\epsilon$ or
  $\max_{k}\epsilon_{k}$, the maximum of $F_{\min}$ over all sets of
  input states of size $N$. Instead of the maximum $F_{\min}$ one can
  seek the minimum $F_{\max}$ or $F_{\max}-F_{\min}$.  We conjecture
  that symmetric POVM states are optimal among all sets of $N$ input
  states, but numerical evidence suggests that symmetric POVMs do not
  exist for many $N$ with $d+1<N<d^2$~\cite{Tropp2005,Sustik2007,Fickus2015}.
  Finally, we observed 
  that a set of input states containing both two MUBs and the 
  simplex states determined $F(\cE)$ for the channel 
  Eq.~\ref{min channel}. This suggests the question of characterizing
  sets of input states and $\bm{\epsilon}$ that together determine the process fidelity exactly.
    
\begin{acknowledgments}
  The authors thank Charles Baldwin, Graeme Smith, Felix Leditzky, 
  and Scott Glancy for helpful conversations and comments on the manuscript.
  This work includes contributions of the
  National Institute of Standards and Technology, which are not
  subject to U.S. copyright. The identification of any product or trade names
  is for informational purposes and does not imply
  endorsement or recommendation by NIST.
\end{acknowledgments}

\appendix

\section{Equivalence of UIC and graph connectivity}

\begin{proposition}\label{prop:uic=gc}
For a set of pure states $\{\ket{\psi_k}\}$, let $G=(V,E)$ be the graph with $V=\{1,\ldots,k\}$ and
$E=\{(k,k'):\bra{\psi_k}\ket{\psi_{k'}}\ne0\}$. The following two conditions are equivalent:\\

1. If $U\in U(d)$ and $[U,\hat{\psi}_k]=0$ for all $k$, then $U\propto I$\\

2. The $\ket{\psi_k}$ span $\mathcal{H}$ and the graph $G$ is connected.
\end{proposition}

\begin{proof}
  (1$\implies$2) This direction is essentially the same
  as the only-if part of Thm.~\ref{theorem1}.
  We prove the contrapositive. Suppose first that
  the states $\ket{\psi_k}$ do not span $\cH$. Let $S$ be the span of
  $\{\ket{\psi_k}\}$, and $S^{\perp}$ the orthogonal complement of
  $S$. Then $U=I_S\oplus i I_{S^{\perp}}$ commutes with all
  $\hat{\psi}_k$ but is not proportional to the
  identity.  Next, suppose that
  $G$ is disconnected. Let $S$ be the span of the states
  $\ket{\psi_k}$ in a connected component of $G$.
  Then $S\not=0$ and $S^{\perp}\not=0$, and again $U=I_S\oplus i
  I_{S^{\perp}}$ commutes with all $\hat{\psi}_k$ but is
  not proportional to the identity.
  
  (2$\implies$1) Suppose that $U$ commutes with all $\hat{\psi}_k$.
  From $U\ket{\psi_k}\bra{\psi_k}=\ket{\psi_k}\bra{\psi_k}U$ it
  follows that $U\ket{\psi_k}=\omega_k\ket{\psi_k}$ with
  $\omega_k=\bra{\psi_k}U\ket{\psi_k}$. Therefore,
  $\braket{\psi_{k'}}{\psi_k}=
  \bra{\psi_{k'}}U^{\dagger}U\ket{\psi_k}=
  \omega_{k'}^*\omega_k\braket{\psi_{k'}}{\psi_k}$.
  If $k$ and $k'$ are adjacent in $G$, we can divide both sides by
  $\braket{\psi_{k'}}{\psi_k}$ and conclude that
  $\omega_k=\omega_{k'}$. Since $G$ is connected, all
  $\omega_k$ are equal, and since the states $\ket{\psi_k}$ span $\cH$,
  it follows that $U$ is proportional to the identity .
\end{proof}

\section{Proof of Thm.~\ref{thm:fid_lowerbnd}}
\label{ap:fid_lowerbnd}

Suppose that the family of input states
$\{\ket{\psi_k}\}_{k=1}^{N}$ spans $\cH$, the graph
$G=(V,E)$ defined by $V=\{1,\ldots,k\}$ and
$E=\{(k,k'):\bra{\psi_k}\ket{\psi_{k'}}\ne0\}$ is connected, and for all $k$ the process $\cE$ satisfies
\begin{equation}\label{ap:constraint}
\quad\bra{\psi_k}\mathcal{E}(\hat{\psi}_k)\ket{\psi_k}\ge 1 - \epsilon.
\end{equation}
By dilation we can express $\cE$ as
\begin{equation*}
\mathcal{E}(\hat{\psi}_k)=\mathrm{Tr}_B(U(\hat{\psi}_k\otimes\ket{0}\bra{0})U^{\dag}),
\end{equation*}
  where $U$ is unitary and we introduced an ancillary system with
  initial state $\ket{0}$.  We label the original input system by $A$,
  the ancillary system by $B$ and disambiguate kets and operators
  with label subscripts and bras with label presuperscripts, when
  necessary. The state $U\ket{\psi_k}\ket{0}$ can be written as
\begin{equation}\label{ap: state decomp}
    U\ket{\psi_k}\ket{0}=a_k\ket{\psi_k}_{A}\ket{e_k}_{B} + b_k\ket{s_k}_{AB},
\end{equation}
where $\ket{e_k}$ is a normalized ancilla state,
and $\ket{s_k}_{AB}$ satisfies ${}^A\!\!\bra{\psi_k}{}\ket{s_k}_{AB}=0$,
and $a_{k}$ and $b_{k}$ are non-negative. The coefficients
  and states can be determined from the identity $a_{k}\ket{e_k}_B =
  {}^{A}\!\!\bra{\psi_{k}}U\ket{\psi_{k}}_{A}\ket{0}_{B}$.  
Eq.~\ref{ap:constraint} implies that $a_k\ge \sqrt{1-\epsilon}$ and
therefore $b_k\le\sqrt{\epsilon}$. Applying Eq.~\ref{ap: state decomp}
  for indices $k$ and $k'$ gives
  \begin{align*}
    \bra{\psi_{k'}}\ket{\psi_k} &=
    \bra{0}\bra{\psi_{k'}}U^{\dagger}U\ket{\psi_k}\ket{0}\notag\\
    &=
    a_{k'}a_{k}\bra{\psi_{k'}}\ket{\psi_k}\bra{e_{k'}}\ket{e_k}\notag\\
    &\hphantom{=\;}+
    b_{k'}a_{k}\bra{s_{k'}}\,\ket{\psi_k}\ket{e_k}+
    a_{k'}b_{k}\bra{\psi_{k'}}\bra{e_{k'}}\,\ket{s_{k}}\notag\\
    &\hphantom{=\;}+O(\epsilon).
  \end{align*}
  Let $\alpha_{k'k} = b_{k'}a_{k}\bra{s_{k'}}\,\ket{\psi_k}\ket{e_k}+
  a_{k'}b_{k}\bra{\psi_{k'}}\bra{e_{k'}}\,\ket{s_{k}}$.  Then
  $|\alpha_{k'k}|\leq 2\sqrt{\epsilon}$, and since
  $a_{k'}a_{k}=1-O(\epsilon)$, we have
  \begin{equation}\label{yet another equation}
    \bra{\psi_{k'}}\ket{\psi_k} =
    \bra{\psi_{k'}}\ket{\psi_k}\bra{e_{k'}}\ket{e_k}+\alpha_{k'k}+O(\epsilon).
  \end{equation}
If $k$ and $k^{\prime}$ are adjacent in
$G$ we can divide both sides of Eq.~\ref{yet another equation} by
$\bra{\psi_{k'}}\ket{\psi_k}$, and obtain
\begin{equation*}
\bra{e_{k'}}\ket{e_k} = 1-\frac{\alpha_{k'k}}{\bra{\psi_{k^{\prime}}}\ket{\psi_k}}+O(\epsilon).
\end{equation*}
If $k$ and $k^{\prime}$ are not adjacent, there is a path $P$ from $k$
to $k^{\prime}$, and the above equation applies for each edge along
the path. We make repeated use of the following fact: if
$\braket{b}{a}=1-\beta_1$ and $\braket{c}{b}=1-\beta_2$, for
$\beta_1$,$\beta_2\in\mathbb{C}$, then $\braket{c}{a}=1-\beta$, with
$\abs{\beta}\le(\sqrt{\abs{\beta_1}}+\sqrt{\abs{\beta_2}})^2$, to
leading order in $\abs{\beta_1}$,$\abs{\beta_2}$. This can be verified
by expanding
$\bra{c}\ket{a}=\bra{c}\left(\ket{b}\bra{b}+\Pi\right)\ket{a}$ with
$\Pi$ the projector onto the orthogonal complement of $\ket{b}$. We
conclude that
\begin{equation}\label{ap:e prime vs e bound}
    \braket{e_{k'}}{e_k}=1-\gamma_{kk'} + O(\epsilon),
\end{equation}
for complex $\gamma_{kk'}$ satisfying
\begin{equation*}
    \abs{\gamma_{kk'}}\le 2\sqrt{\epsilon}\, \bigg( \sum_{(i,j)\in P}\abs{\braket{\psi_i}{\psi_j}}^{-1/2} \bigg)^{2}.
\end{equation*}
Because this is true for any path from $k$ to $k^{\prime}$, we can choose the path such that the above sum is minimized. Therefore,
\begin{equation}\label{ap:env bounds}
\abs{\gamma_{kk'}}\le 2\sqrt{\epsilon}\,W_{kk'}^2
\end{equation}
where $W_{k k^{\prime}}$ is defined by Eq.~\ref{weight}.

To compute the process fidelity we add an additional system $C$ and start with
$AC$ in the maximally entangled state
$\ket{\phi}_{AC}=\frac{1}{\sqrt{d}}\sum_x\ket{x}_A\ket{x}_C$.
The process fidelity is then given by
\begin{equation}  \label{ap:phi_fe}
  F(\cE) = \Tr_{B}\bra{\phi}_{AC}U_{AB}\ket{\phi}_{AC}\ket{0}_B
  \bra{0}_{B}\bra{\phi}_{AC}U_{AB}^{\dagger}\ket{\phi}_{AC}.
\end{equation}
  By reordering if necessary, we can assume that
  $\{\ket{\psi_{k}}\}_{k=1}^{d}$ is a basis. There exists a (non-orthogonal
and un-normalized) dual basis
$\{|\widetilde{\psi}_k\rangle\}_{k=1}^{d}$, satisfying
$\bra*{\widetilde{\psi}_k}\ket*{\psi_{k'}}=\delta_{kk'}$ for $1\leq
k,k'\leq d$. For the remainder of this proof, indices $k,k'$ are
  in $\{1,\ldots,d\}$ by default. The computational basis states can be expanded as
\begin{equation*}
    \ket{x}=\sum_k\bra*{\widetilde{\psi}_k}\ket*{x}\ket{\psi_k}.
\end{equation*}
Expanding $\ket{\phi}_{AB}$ in terms of the computational
basis and invoking Eq.~\ref{ap: state decomp} gives
\begin{align*}
\begin{split}
    U_{AB}\ket{\phi}_{AC}\ket{0}_B=\frac{1}{\sqrt{d}}\sum_x\sum_k \bra*{\widetilde{\psi}_k}\ket*{x} U_{AB}\ket{\psi_k}_A\ket{0}_B\ket{x}_C\\
    =\frac{1}{\sqrt{d}}\sum_x\sum_k \bra*{\widetilde{\psi}_k}\ket*{x}\big(a_k\ket{\psi_k}_A\ket{e_k}_B+b_k\ket{s_k}_{AB}\big)\ket{x}_C.
    \end{split}
\end{align*}
Applying $\bra{\phi}_{AC}$ on the left gives
\begin{align*}
\begin{split}
    \frac{1}{d}&\sum_x\sum_k \braket*{\widetilde{\psi}_k}{x}\big(a_k\braket{x}{\psi_k}\ket{e_k}_B+b_k{}^{A}\!\!\bra*{x}{\!\!}\ket*{s_k}_{AB}\big)\\
    &=\frac{1}{d}\sum_k\big( a_k\ket{e_k}_B+b_k{}^{A}\!\!\bra*{\widetilde{\psi}_{k}}{\!\!}\ket*{s_k}_{AB} \big).
    \end{split}
\end{align*}
Substituting in Eq.~\ref{ap:phi_fe} yields
\begin{multline}
    F(\cE)= \frac{1}{d^2}\sum_{kk'} \Big\{\braket{e_{k'}}{e_k} + b_k{}^{A}\!\!\bra*{\widetilde{\psi}_{k}}{}^{B}\!\!\bra{e_{k'}}{\!\!}\ket{s_{k}}_{AB}+\\
    b_k^*{}^{AB}\!\!\bra*{s_{k'}}{\!\!}\ket*{\widetilde{\psi}_{k'}}_A\ket{e_k}_B\Big\}+O(\epsilon). \label{ap:fidexpr}
\end{multline}
To bound the magnitude of the sum involving
  $\braket{e_{k'}}{e_k}$, we apply Eqs.~\ref{ap:e prime vs e bound}
  and~\ref{ap:env bounds} to obtain
\begin{align*}
\begin{split}
    \abs{\frac{1}{d^2}\sum_{kk'}\braket{e_{k'}}{e_k}} &= \abs{\frac{1}{d}+\frac{1}{d^2}\sum_{k\ne k'}(1-\gamma_{kk'})}\\
    &=\abs{1-\frac{1}{d^2}\sum_{k\ne k'}\gamma_{kk'}}\\
    &\ge 1-\frac{4\sqrt{\epsilon}}{d^2}\sum_{k>k'}W_{kk'}^2
    \end{split}
\end{align*}
The terms
${}^{A}\!\!\bra*{\widetilde{\psi}_{k}}{}^{B}\!\!\bra*{e_{k'}}{\!\!}\ket*{s_k}_{AB}$
and ${}^{AB}\!\!\bra*{s_{k'}}{\!\!}\ket*{\psi_{k'}}_{A}\ket*{e_k}_{B}$
are each bounded in magnitude by $\|\ket*{\widetilde{\psi}_k}\|$.
To express this quantity in terms of the $\ket{\psi_{k}}$, define
  $C=\sum_k\ket{\psi_k}\bra{k}$.  Since
  $\bra{k}C^{-1}\ket{\psi_{k'}}=\bra{k}C^{-1}C\ket{k'}=\delta_{kk'}$
  for all $k,k'$, we have $\bra*{\widetilde{\psi}_{k}}=\bra{k}C^{-1}$
  and $\| \ket*{\widetilde{\psi}_{k}}\|^{2}
  =\bra{k}C^{-1}(C^{-1})^{\dagger}\ket{k}=\bra{k}(C^{\dagger}C)^{-1}\ket{k}$.
  The matrix $M=C^{\dagger}C$ can be recognized as the Gram matrix for
  the states $\ket{\psi_{k}}$, in terms of which we can write
  \begin{equation*}
    \|\ket*{\widetilde{\psi}_k}\|=\sqrt{(M^{-1})_{kk}}\,.
  \end{equation*}

Substituting these bounds into the
expression for the process fidelity in Eq.~\ref{ap:fidexpr} gives
\begin{equation*}
F(\cE)\ge 1-\frac{2}{d}\Big(\frac{2}{d}\sum_{k> k'}W_{kk'}^2+\sum_k\sqrt{(M^{-1})_{kk}}\Big)\sqrt{\epsilon}+O(\epsilon),
\end{equation*}
matching Thm.~\ref{thm:fid_lowerbnd} in the main text. This
lower bound can be generalized to the case of state dependent errors
$\epsilon_k$. Working back through the derivation, it suffices to
apply the following replacements to the expression for the lower bound:
\begin{equation*}
W_{kk'}^2\sqrt{\epsilon}\, \mapsto\, \min_{P\in\mathcal{P}_{kk'}}\Bigg(\sum_{(i,j)\in P}\sqrt{\frac{\sqrt{\epsilon_i}+\sqrt{\epsilon_j}}{\abs{\braket{\psi_i}{\psi_j}}}}\Bigg)^2,
\end{equation*}
and
\begin{equation*}
\sum_k\sqrt{(M^{-1})_{kk}}\sqrt{\epsilon}\,\mapsto\,\sum_k\sqrt{(M^{-1})_{kk}}\sqrt{\epsilon_{k}}.
\end{equation*}

\section{Derivation of Eq.~\ref{min chan hoff bounds}}

We set $N=d+1$. Given the channel
\begin{align*}
\begin{split}
    \mathcal{E}(\rho)&=(1-p)\rho+p\frac{d}{d+1}\sum_{k=0}^{d}\hat{\psi}_k\rho\,\hat{\psi}_k,\\
    p&=\frac{d^2}{d-1}\epsilon, 
    \end{split}
\end{align*}
and the expression for $\ket{\psi_k}$ in Eq.~\ref{def psi k}, we compute
\begin{align*}
\begin{split}
    \bra{x}\mathcal{E}(\hat{x})\ket{x} &= 1-p+p\frac{d}{d+1}\sum_{k=0}^{d}\abs{\bra{x}\ket{\psi_k}}^4\\
    &= 1-d\,\epsilon.
    \end{split}
\end{align*}
Therefore,
\begin{equation*}
F_1=\frac{1}{d}\sum_x\bra{x}\mathcal{E}(\hat{x})\ket{x} = 1-d\,\epsilon.
\end{equation*}
For the Fourier basis,
\begin{equation*}
    \bra{f_x}\mathcal{E}(\hat{f}_x)\ket{f_x}=1-p+p\frac{d}{d+1}\sum_{k=0}^{d}\abs{\bra{f_x}\ket{\psi_k}}^4,
\end{equation*}
where
\begin{equation*}
    \bra{f_x}\ket{\psi_k}=\frac{1}{d}\sum_{y=0}^{d-1}e^{2\pi i y(\frac{k}{d+1}-\frac{x}{d})}.
\end{equation*}
 We compute
\begin{equation*}
\sum_{k=0}^{d}\abs{\bra{f_{x}}\ket{\psi_k}}^4=\frac{d+1}{d^4}\sum_{(y_1,y_2,y_3,y_4)\in J}e^{-2\pi i(y_1-y_2+y_3-y_4)x/d},
\end{equation*}
where $J$ consists of the tuples $(y_{1},\ldots,y_{4})$ satisfying
  $y_1-y_2+y_3-y_4=0 \mod (d+1$) and $y_{i}\in\{0,\ldots,d-1\}$.
  For $m=0,d+1,-(d+1)$, let $J_{m}$ be the set of tuples
  $(y_{i})_{i=1}^{4}\in J$ such that $y_1-y_2+y_3-y_4=m$.
  Define $S_{m}(x)=\sum_{(y_1,y_2,y_3,y_4)\in J_{m}}e^{-2\pi i(y_1-y_2+y_3-y_4)x/d}$.
  Then $\sum_{x}S_{\pm (d+1)}(x) = 0$ and 
  $\sum_{x}S_{0}(x) = d |J_{0}|$. For $|J_{0}|$ we get
  \begin{align}
    |J_{0}| &= \sum_{l=0}^{2(d-1)}\Big|\big\{(y_{1},y_{3}):\notag\\
    &\hphantom{\sum_{l=0}^{2(d-1)}|\{}y_{1}+y_{3}=l\textrm{\ and\ }
    y_{1},y_{3}\in\{0,\ldots,d-1\}\big\}\Big|^{2}\notag\\
    &= \sum_{l=0}^{d-1}(l+1)^{2}+\sum_{l=d}^{2(d-1)}(2(d-1)-l+1)^{2}\notag\\
    &= d^{2}+2 \sum_{l=1}^{d-1}l^{2}\notag\\
    &= d^{2} + 2\frac{1}{6}(d-1)(d)(2d-1).\notag
  \end{align}
  We can now evaluate $F_{2}$.
  \begin{align}
    F_2 &= \frac{1}{d}\sum_x\bra{f_x}\mathcal{E}(\hat{f}_x)\ket{f_x}\notag\\
        &= 1-\left(\frac{d^{2}}{d-1}
        - \frac{d^{3}}{(d-1)(d+1)}\frac{d+1}{d^{4}}|J_{0}|\right)\epsilon\notag\\
      &= 1-\left(\frac{d^{2}}{d-1}
        - \frac{1}{d(d-1)}|J_{0}|\right)\epsilon\notag\\
      &= 1-\left(\frac{d^{2}}{d-1}
        - \frac{d}{d-1}-\frac{2d-1}{3}\right)\epsilon\notag\\
      &=1- \frac{d+1}{3}\epsilon.\notag
  \end{align}

%\bibliographystyle{unsrt}
%\bibliography{library}
%\input{main.bbl}

\end{document}